\documentclass{article}

\usepackage{PRIMEarxiv}

\usepackage[utf8]{inputenc} 
\usepackage[T1]{fontenc}    
\usepackage{hyperref}       
\usepackage{url}            
\usepackage{booktabs}       
\usepackage{amsfonts}       
\usepackage{nicefrac}       
\usepackage{microtype}      
\usepackage{lipsum}
\usepackage{amsmath,amssymb,amsthm}
\usepackage{fancyhdr}       
\usepackage{graphicx}       
\usepackage{listings}
\graphicspath{{media/}}     

\title{Matrix Based Adaptive Short Block Cipher
\thanks{\textit{\underline{Citation}}: 
\textbf{Authors. Title. Pages.... DOI:000000/11111.}} 
}

\author{
  Awnon Bhowmik\\
  Department of Mathematics\\
  University of Central Florida\\
  \texttt{awnon.bhowmik@ucf.edu}
}

\newtheorem{thm}{Theorem}
\theoremstyle{definition}

\begin{document}
\maketitle

\begin{abstract}
Every day, millions of credit cards are swiped and transactions are carried out across the world. Due to numerous forms of unethical digital activities, users are vulnerable to credit card fraud, phishing, identity theft, etc. This paper outlines a novel block encryption algorithm involving multiple private keys and a resilient trapdoor function that ensures data security while maintaining an optimal run time and space complexity. The proposed scheme consists of an irrepressible trapdoor based on a depressed cubic function and a unique key generation algorithm that uses Fibonacci sequences and invertible square matrices for improved security. The paper involves data obtained from comprehensive crypt analysis exploiting the strengths and weaknesses of the system and comments on its potential large-scale industry applications.
\end{abstract}


\section{Introduction}
Over the years, the new cryptosystems have gradually increased in obscurity thus providing more diffusion and confusion. The essence of any cryptosystem relies on some special mathematical trapdoor function that makes it practically impossible for an unwelcome interceptor to gain access to secretive information. This has brought about better security protocols that are being used in everyday digital systems today. The wonderful fact is that with time it has been realized that any mathematical concept can be used to create trapdoors for cryptosystems as long as they are used appropriately. These functions also ensure that the authorized parties (who know the secret key) can continue sharing data among themselves. The rate of data exchange in the present era requires trapdoor functions that are extremely fast and take up the minimum memory resources while operating on data.

Credit card encryption is a security measure used to minimize the likelihood of card information being stolen as it may lead to countless fraudulent practices \cite{whitworth2001fraud}. In the past, banks used the magnetic strip on the back of credit cards to accept purchases. However, these security features became heavily corrupted with malware during the fraud crisis of 2014. Once the information is stolen, it can be put into a new strip and used to make fraudulent purchases. Additionally, magnets can tamper with the magnetic strip on credit cards, erasing the information stored there and rendering them useless \cite{oldenkamp1982method}. During an ordinary credit card transaction, there are mainly 3 steps where the
card is vulnerable:
\begin{enumerate}
    \item The electronic terminal where the card is being scanned.
    \item The transmission of information between that terminal and its bank’s back-end servers.
    \item The card itself, as flaws in internal components of the card, such as electronic chip
\end{enumerate}

\section{Linear Algebra in Cryptography}
Modern cryptography began only in the 1970s, and over the years many fascinating cryptosystems have been developed by involving it with the various mathematical concepts we are aware of, predominantly by Number Theory since this is where we get to learn about prime numbers and cryptography stands on the foundation of prime number distribution. Anyone finding the pattern to the distribution of primes can erase the concept of cryptography
from textbooks and modern times forever. Not far behind number theory; linear algebra concepts are no slouch themselves. Ranging from the connection
between the determinant and invertibility of a matrix to eigenvalues and eigenvectors as components of potential trapdoor functions, linear algebra provides a wide range of tools up for grabs \cite{staffelbach1990cryptographic}. The following talks about how
it is used in cryptography.

\section{Encryption in credit card transactions}
Credit cards constitute a straightforward and fast method of digital payment scheme that allows the transfer of massive proportions of money of various currencies and denominations. This makes credit cards an intrinsic component of the payment process. Businesses generally have electronic terminals where consumers can simply swipe or scan the cards and the appropriate amount of money gets seamlessly transferred. The card has unique identification digits stored on the onboard chip that gets sent to authorized computer servers where the balance amount is re-evaluated after deducting the appropriate amount and verification is carried out. There are a couple of mechanisms used by banks to ensure secure transactions. Some of these include a private PIN number for authentication, CVV numbers for online transactions and even encrypting the magnetic strip at the back of the card. All these methods are effective, however, there is still a need for a more efficient encryption scheme to ensure that phishing, identity theft, and other forms of online money-based fraud are. Credit card encryption makes it impossible for an unethical hacker to access the card credentials without a special private key that ensures seamless transactions between various authorized parties. As it turns out, the proposed algorithm was observed to deliver superior performance in terms of security and run time for shorter string lengths. This attribute makes this system a feasible candidate for credit card encryption that could with multiple applications in different stages of the process.

\subsection{Summary of applications in Linear Algebra}
\begin{enumerate}
    \item A matrix can be used as a cipher to encrypt a message.
    \begin{itemize}
        \item The matrix must be invertible for use in decrypting.
    \end{itemize}
    
    \item Cipher matrix can be as simple as a $3\times 3$ matrix composed of random integers.
    
    \item To encrypt plaintext, each character in the plaintext must be denoted with a numerical value and placed into a matrix.
    
    \begin{itemize}
        \item These numbers can range in value, but an example is using $1-26$ to represent A to Z and $27$ to represent a space.
    \end{itemize}
    
    \item This matrix is then multiplied with the cipher matrix to form a new matrix containing the ciphertext message.
\end{enumerate}

\subsection{Encrypting a message}
\begin{enumerate}
    \item Each character of the message is given a numerical value as mentioned above. These values are then separated into vectors, such that the number of rows of each vector is equivalent to the number of rows of the cipher matrix.
    \item Values are placed into each vector one at a time, going down a row for each value. Once a vector is filled the next vector is created. If the last vector does not get filled by the plaintext then the remaining entries will hold the value for a space.
    \item The vectors are then augmented to form a matrix that contains the plaintext.
    \item The plaintext matrix is then multiplied with the cipher matrix to create the ciphertext matrix
\end{enumerate}

\subsection{Decrypting a message}
\begin{enumerate}
    \item To decrypt a ciphertext matrix the original cipher matrix must be used. The cipher matrix must be inverted to decrypt the ciphertext. This cipher matrix can be used as the private key for decryption.
    \item This inverted cipher matrix is then multiplied with the ciphertext matrix.
    \item The product produces the original plaintext matrix.
    \item The plaintext can be found again by taking this product and splitting it back up into its separate vectors and then converting the numbers back into their corresponding alphabetic form.
\end{enumerate}

\section{Confusion and Diffusion}
In cryptography, confusion and diffusion are two properties of the operation of a secure cipher which were identified by Claude Shannon in his paper Communication Theory of Secrecy Systems, published in 1949 \cite{shannon1949communication}. In Shannon’s original definitions, confusion refers to making the relationship between the key and the ciphertext as complex and involved as possible \cite{stallings2006cryptography}; diffusion refers to the property that the redundancy in the statistics of the plaintext is "dissipated" in the statistics of the ciphertext. In other words, the non-uniformity in the distribution of the individual letters (and pairs of neighboring letters) in the plaintext should be redistributed into the non-uniformity in the distribution of much larger structures of the ciphertext, which is much harder to detect. In diffusion, the statistical structure of the plaintext is dissipated into long-range statistics of the ciphertext. Diffusion means that the output bits should depend on the input bits in a very complex way. In a cipher with good diffusion, if one bit of the plaintext is changed, then the ciphertext should change completely, in an unpredictable
or pseudorandom manner. AES protocol is a perfect example that exhibits both confusion and diffusion.

\section{Trapdoor Function}
The essence of any cryptosystem relies on some special mathematical trapdoor function \cite{yao1982theory} that makes it practically impossible for an unwelcome interceptor to gain access to secretive information. Simultaneously, these functions also ensure that the authorized parties (who know the secret key) can continue sharing data among themselves. A trapdoor function is a mathematical transformation that is easy to compute in one direction, but extremely difficult (practically impossible) to compute in the opposite direction in a feasible time unless some special information is known (private key). Analogously, this can be thought of as the lock and key in modern cryptography where until and unless someone has access to the exact key, they can’t open the lock. In mathematical terms, if $f$ is a trapdoor function, then $y=f(x)$ easy to calculate but $x=f^{-1}(y)$ is tremendously hard to compute without some special knowledge $k$ (called key). In case $k$ is known, it becomes easy to compute the inverse $x=f^{-1}(y,k)$.

\section{Initial Encoding Procedure}
\label{enc}
The encoding procedure depends on a few results which are modified according to fit the needs of this paper.\\~\\
\textbf{Cantor Pairing Function} The Cantor Pairing Function is a primitive recursive pairing function $\pi:\mathbb{N}\times\mathbb{N}\to\mathbb{N}$ defined by \begin{equation}
\label{eq:1}
    \pi(k_1,k_2)=\dfrac{1}{2}(k_1+k_2)(k_1+k_2+1)+k_2
\end{equation}
\begin{thm}
\label{thm:1}
\normalfont
    Product of $n$ consecutive integers is divisible by $n$.
\end{thm}
\begin{proof}
    By contradiction. Suppose that none of the $n$ consecutive integers is divisible by $n$. There are $n-1$ different remainders upon division by = while $n$ remainders are needed. Therefore, by the pigeonhole principle, two of the divisions have the same remainder. But that would imply that the difference in those dividends is divisible by $n$ and hence is $n$. But the largest difference is $n-1$. Contradiction.
\end{proof}

The algorithm starts off with a list of ASCII values and a list of unique primes. The process is as follows.\\~\\
\textbf{Encode}
\begin{enumerate}
    \item Let $x$ be the ASCII value and $y$ is a prime. Compute \begin{equation}
    \label{eq:encode}
        t=\dfrac{(x+y-1)(x+y)(x+y+1)}{6}
    \end{equation}
    \item Setting $n=x+y$ yields $t=\dfrac{n^3-n}{6}$ and $n^3-n\equiv 0\pmod{6}$ by Theorem \ref{thm:1}.
\end{enumerate}

\textbf{Decode}
\begin{enumerate}
    \item Given $t$, we solve $n^3-n-6t=0$ for $n$.
    \item The prime number $y$ is known, so $x=n-y$ returns the ASCII data.
\end{enumerate}

\section{Remarks on the cubic equation}
A cubic equation\begin{equation}
    ax^3+bx^2+cx+d=0
\end{equation} can be converted to a \emph{depressed} cubic equation via small substitutions to attain the form of \begin{equation}
    x^3+mx+n=0
\end{equation}. For this equation, it is required that $27n^2+4m^3\neq 0$ to avoid singularities. In our case, $n^3-n-6t=0$, we require

\begin{align*}
    27(-6t)^2+4(-1)^3&\neq 0\\
    27\cdot 36t^2-4&\neq 0\\
    243t^2-1&\neq 0\\
    t&\neq\dfrac{\sqrt{3}}{27}
\end{align*}

Since $t\in\mathbb{Z}$, this is impossible.

\section{An example demonstrating the initial encoding procedure}
Suppose the input character is 'A' which has an ASCII value of $65$, and the corresponding prime is $13$. The following calculations take place

\emph{Encode}\\
$$t=\dfrac{1}{6}(65+13-1)(65+13)(65+13+1)=79079$$\\
\emph{Decode}\\
\begin{align*}
    n^3-n-6t&=0\\
    n^3-n-6(79079)&=0\\
    n^3-n-474474&=0
\end{align*}
Let $n=\alpha+\beta$, then
\begin{align*}
    n^3&=(\alpha+\beta)^3\\
    &=\alpha^3+\beta^3+3\alpha\beta(\alpha+\beta)\\
    &=\alpha^3+\beta^3+3\alpha\beta n\\
    n^3-3\alpha\beta n-(\alpha^3+\beta^3)&=0
\end{align*}

Comparing coefficients, we have
\begin{equation*}
    \begin{split}
        3\alpha\beta&=1\\
        \alpha\beta&=\dfrac{1}{3}
    \end{split}
    \qquad
    \begin{split}
        \alpha^3+\beta^3&=474474
    \end{split}
\end{equation*}

Generating a quadratic equation with the roots $\alpha^3$ and $\beta^3$
\begin{align*}
    s^2-(\alpha^3+\beta^3)s+\alpha^3\beta^3&=0\\
    s^2-474474s+\dfrac{1}{27}&=0\\
    s&=\dfrac{474474\pm\sqrt{474474^2-\dfrac{4}{27}}}{2}\\
    s&=(\alpha^3,\beta^3)=474474,7.9\times10^{-8}\\
    n&=\alpha+\beta=\sqrt[3]{474474}+\sqrt[3]{7.9\times 10^{-8}}=78
\end{align*}

Here is a Mathematica notebook demonstration.\\
\begin{lstlisting}[language=Mathematica]
    In: x:=65
        y:=13
        n:=x+y
        t:=(n^3-n)/6
        Solve[r^3-r-6t==0,r,Integers]
    Out:{{r->78}}

    In: 78-y
    Out: 65

    In: ClearAll["Global`*"]
    In: f[x_]:=x^3-x-6t
    In: f'[x]
    Out: -1+3x^2
    In: Reduce[f'[x]>0,x]
    Out:
\end{lstlisting}
$$x<-\frac{1}{\sqrt{3}}\lor x>\frac{1}{\sqrt{3}}$$

This implies that the notebook demonstration also showed that a given function $$f(x)=x^3-x-6t$$ exhibits an increasing behavior on $x\in\left(\dfrac{1}{\sqrt{3}},\infty\right)$. Since $x,y>0$, and $n=x+y$, we have $n>0$, and so $f(n)=n^3-n-6t$ will always be increasing on $\mathbb{Z}^+$.

\section{Proposed algorithm}
\subsection{Encryption}
\begin{enumerate}
    \item Input a message string
    \item Generate an array $A$ of ASCII values and another array $P$ containing unique pseudorandom primes.
    \item Take the ASCII values and their corresponding mapped prime and apply the encoding scheme described in section \ref{enc}.
    \item Put these values into matrix $B$ and split them into blocks of $4$. Apply '$0$' padding to whichever $1\times 4$ array isn't completely filled with entries. Convert each $1\times 4$ array into a $2\times 2$ matrix. Suppose each block matrix is $B_i$.
    \item Generate a $2\times 2$ Fibonacci matrix $Q^n$ \cite{gould1981history}, where $$Q^n=\begin{bmatrix}
        F_{n+1} & F_n\\F_n & F_{n-1}
    \end{bmatrix}$$ and perform $$C=B\cdot Q^n$$

    \item Choose an integer $n$ such that $\theta=\dfrac{n\pi}{2}$ to generate a rotation matrix $R$ such that $$R=\begin{bmatrix}
        \cos\theta & -\sin\theta\\\sin\theta &\cos\theta
    \end{bmatrix}$$ and perform
    \begin{align*}
        D&=C\cdot R\\
        D&=D^T
    \end{align*}
    \item Choose a $2\times 2$ \emph{invertible} secret key $K$ and perform encryption $$E=D\cdot K$$
\end{enumerate}

After the final step, the result will alter drastically depending on the contents of the original string, thus demonstrating \emph{confusion}. On the other hand, the choice of a different key for the same message string will
completely alter the result, thus demonstrating \emph{diffusion}.

\subsection{Decryption}
\begin{enumerate}
    \item Perform $D=E\cdot K^{-1}$
    \item Perform transpose $(D^T)^T=D$
    \item Perform $C=D\cdot R^{-1}$
    \item Perform $B=C\cdot (Q^n)^{-1}$
    \item Apply the decoding scheme described in section \ref{enc} to generate the ASCII data.
    \item Convert each $2\times 2$ block to $1\times 4$ array. Remove padding.
\end{enumerate}

This system can probably be extended to a $n\times n$ block size with numerous tweaks in the proposed algorithm while dealing with exponential runtime complexity.

\section{Properties of determinants}
Determinants play an important part in this cryptosystem so it is essential to recall their properties. For any arbitrary square matrix $A$
\begin{enumerate}
    \item $\det(A)=\det(A^T)$
    \item For any arbitrary constant $c$, $\det(cA)=c^n\det(A)$
\end{enumerate}

\section{Few remarks on R and K}
The $2D$ rotation matrix is given by $$R=\begin{bmatrix}
    \cos\theta & -\sin\theta\\\sin\theta &\cos\theta
\end{bmatrix}$$ It is to be noted that $det(R)\neq 0$. It means the sender can choose any $\theta\in\mathbb{R}$ to encrypt the data. The only drawback is that the sine and cosine ratios can provide only three integral values from $\left\{-1,0,1\right\}$. Any other angles will raise floating point precision errors which will amount to incorrect decryption of the data.

A rotation matrix is a square matrix with real entries, and also an orthogonal matrix with a determinant of $1$. A matrix $R$ is orthogonal iff $$R^TR=RR^T=I$$ which consequently means that $$R^{-1}=R^T$$This can make things a bit simpler while coding up the algorithm.\\~\\
An arbitrarily chosen matrix $K$ should be checked for invertibility since it is used in the decryption procedure. There can be a \textit{special case} where $K$ itself is a rotation matrix. Then $$K^{-1}=K^T$$ can speed up the process.\\
It is not necessary to choose $K$ as a square matrix at all. Suppose we choose a $m\times n$ matrix $M$, so $M^T$ would be of order $n\times m$. The product $M^TM$ will be a square symmetric matrix of order $m\times m$.

If $M$ is chosen as skew-symmetric, then
\begin{align*}
    M&=M^T\qquad\text{By property 1}\\
    \det(M)&=\det(M^T)\\
    \det(M)&=\det(-M)\\
    \det(M)&=(-1)^n\det(M)\\
    \det(M)&=\det(M)\qquad\text{for odd }n\\
    2\det(M)&=0\\
    \det(M)&=0
\end{align*}
Hence, for an odd-order square matrix, $M$ shouldn’t be chosen as skew symmetric to avoid the singularity.

\section{Automated generation of key matrix K}
The user inputting a $2\times 2$ matrix that has to always account for its invertibility by calculating its determinant. A singular matrix has a determinant of $0$ and this must be avoided since $K^{-1}$ is required during decryption. To overcome this hassle, a random $2\times 2$ matrix with integer entries can be checked for non-singularity by checking
its column independence.\\

Given an arbitrary matrix with integer entries
$$K=\begin{bmatrix}
    a & b\\ c & d
\end{bmatrix}$$
For arbitrary scalars $k_1,k_2$, we are required to solve the following system
$$k\begin{bmatrix}
    a\\c
\end{bmatrix}+k_2\begin{bmatrix}
    b\\d
\end{bmatrix}=\begin{bmatrix}
    0\\0
\end{bmatrix}$$
which is equivalent to the linear system $$\begin{cases}k_1a+k_2b&=0\\k_1c+k_2d&=0\end{cases}$$\\
If $k_1=k_2=0$, the matrix $K$ is column independent and thereby non-singular.

\section{Conclusion}
This paper introduced a new encryption scheme that was inspired by the Hill Cipher. A couple of edge cases were disclosed and further study is required to make this system even better. The runtime complexity can be safely assumed to be exponential or that it would reach exponential complexity as the length of the string increases or the other parameters are varied. Currently, this system may not be perfect for use in the industry but it is a system nevertheless. Further study might give insights into other problems or drawbacks of the system and appropriate countermeasures can be taken to make it good enough to be used in the industry. Our work so far has been compiled into a GitHub repository available here \cite{bhowmik_menon_2020}.

\bibliographystyle{unsrt}  
\bibliography{references}  

\begin{thebibliography}{1}

\bibitem{whitworth2001fraud}
Brian Whitworth.
\newblock Fraud resistant credit card using encryption, encrypted cards on
  computing devices, October~25 2001.
\newblock US Patent App. 09/764,369.

\bibitem{oldenkamp1982method}
Ralph~J Oldenkamp.
\newblock Method for tamper-proofing magnetic stripe card reader, March~30
  1982.
\newblock US Patent 4,322,613.

\bibitem{staffelbach1990cryptographic}
Othmar Staffelbach and Willi Meier.
\newblock Cryptographic significance of the carry for ciphers based on integer
  addition.
\newblock In {\em Conference on the Theory and Application of Cryptography},
  pages 602--614. Springer, 1990.

\bibitem{shannon1949communication}
Claude~E Shannon.
\newblock Communication theory of secrecy systems.
\newblock {\em The Bell system technical journal}, 28(4):656--715, 1949.

\bibitem{stallings2006cryptography}
William Stallings.
\newblock {\em Cryptography and network security, 4/E}.
\newblock Pearson Education India, 2006.

\bibitem{yao1982theory}
Andrew~C Yao.
\newblock Theory and application of trapdoor functions.
\newblock In {\em 23rd Annual Symposium on Foundations of Computer Science
  (SFCS 1982)}, pages 80--91. IEEE, 1982.

\bibitem{gould1981history}
Henry~W Gould.
\newblock A history of the fibonacci q-matrix and a higher-dimensional problem.
\newblock {\em Fibonacci Quart}, 19(3):250--257, 1981.

\bibitem{bhowmik_menon_2020}
Awnon Bhowmik and Unnikrishnan Menon.
\newblock
  https://github.com/awnonbhowmik/adaptive-matrix-based-short-block-cipher, Aug
  2020.

\end{thebibliography}

\end{document}